\documentclass[11pt]{article}

\usepackage[utf8]{inputenc}
\usepackage[T1]{fontenc}
\usepackage{iftex}
\PassOptionsToPackage{colorlinks,urlcolor=blue,citecolor=blue,linkcolor=blue}{hyperref}
\usepackage{url}
\usepackage{booktabs}
\usepackage{amsfonts}
\ifdefined\DisableMicrotype
\else
  \ifPDFTeX
    \ifnum\pdfoutput>0
      \usepackage[protrusion=true,expansion=true,tracking=true]{microtype}
    \else
      \usepackage[protrusion=true,expansion=false,tracking=true]{microtype}
    \fi
  \else
    \usepackage[protrusion=true,expansion=false,tracking=false]{microtype}
  \fi
\fi
\usepackage{nicefrac}
\ifPDFTeX
  \usepackage{bbm}
\else

\fi
\usepackage{xspace}
\usepackage{algpseudocode}
\usepackage{algorithm}
\usepackage{bm}
\usepackage{amsmath}
\usepackage{amsthm}
\usepackage{amssymb}
\usepackage{thmtools}
\usepackage{thm-restate}
\usepackage{float}
\usepackage{subcaption}
\usepackage{mathtools}
\usepackage[most]{tcolorbox}
\usepackage{xparse}
\usepackage{paralist}
\usepackage{enumitem}
\usepackage{fancyhdr}
\usepackage{multirow}
\usepackage{pifont}
\usepackage{hyperref}
\usepackage{tikz}
\usepackage{graphicx}

\usetikzlibrary{calc,arrows.meta,positioning}

\usepackage[capitalize,nameinlink,noabbrev]{cleveref}
\newif\ifneuripsbuild
\ifdefined\NeuripsBuild
  \neuripsbuildtrue
\else
  \neuripsbuildfalse
\fi

\newif\ifusedeferredproofs
\ifdefined\EnableDeferredProofs
  \usedeferredproofstrue
\else
  \usedeferredproofsfalse
\fi
\ifusedeferredproofs
  \usepackage[
    createShortEnv,
    commandRef=Cref,
    disablePatchSection
  ]{proof-at-the-end}
\fi
\usetikzlibrary{arrows.meta,positioning,shapes.geometric}

\makeatletter
\@ifundefined{theHALG@line}
  {\def\theHALG@line{\thealgorithm.\arabic{ALG@line}}}
  {\renewcommand{\theHALG@line}{\thealgorithm.\arabic{ALG@line}}}
\makeatother

\DeclareMathOperator{\supp}{supp}

\newcommand{\Z}{\mathbb{Z}}

\newcommand{\poly}{\operatorname{poly}}
\newif\ifshowclarify
\showclarifytrue
\newcommand{\HideClarify}{\showclarifyfalse}

\newif\ifshowchanges
\showchangestrue

\newcommand{\ShowChanges}{\showchangestrue}

\NewDocumentEnvironment{changedblock}{+b}{\ifshowchanges
    \begingroup\color{blue}#1\endgroup
  \else
    #1
  \fi
}{}

\newif\ifshowresearchboxes
\newif\ifshowmotivationboxes
\newif\ifshowinterpretationboxes
\newif\ifshowverificationboxes

\showresearchboxestrue
\showmotivationboxestrue
\showinterpretationboxestrue
\showverificationboxestrue

\newcommand{\HideAllResearchBoxes}{\showresearchboxesfalse
  \showmotivationboxesfalse
  \showinterpretationboxesfalse
  \showverificationboxesfalse
}

\NewDocumentEnvironment{motivation}{+b}{\ifshowresearchboxes
    \ifshowmotivationboxes
      \begin{tcolorbox}[
        breakable,
        colback=blue!4,
        colframe=blue!45!black,
        title=Motivation,
        fonttitle=\bfseries
      ]
      #1
      \end{tcolorbox}
    \fi
  \fi
}{}

\NewDocumentEnvironment{interpretation}{+b}{\ifshowresearchboxes
    \ifshowinterpretationboxes
      \begin{tcolorbox}[
        breakable,
        colback=green!4,
        colframe=green!35!black,
        title=Interpretation,
        fonttitle=\bfseries
      ]
      #1
      \end{tcolorbox}
    \fi
  \fi
}{}

\NewDocumentEnvironment{verification}{+b}{\ifshowresearchboxes
    \ifshowverificationboxes
      \begin{tcolorbox}[
        breakable,
        colback=gray!5,
        colframe=gray!50,
        title=Verification,
        fonttitle=\bfseries
      ]
      #1
      \end{tcolorbox}
    \fi
  \fi
}{}

\theoremstyle{plain}
\numberwithin{equation}{section}
\newtheorem{theorem}{Theorem}[section]
\numberwithin{theorem}{section}

\newtheorem{lemma}[theorem]{Lemma}
\newtheorem{proposition}[theorem]{Proposition}

\theoremstyle{definition}
\newtheorem{definition}[theorem]{Definition}
\newtheorem{remark}[theorem]{Remark}

\newcounter{deferredproofanchor}[section]

\crefalias{deferredproofanchor}{section}

\newcommand{\DeferredProofNotice}[1]{\par\noindent
  \ifusedeferredproofs
    Proof deferred to \hyperref[#1]{\Cref*{#1} (\nameref*{#1})}.\else
    Proof deferred to \hyperref[#1]{\Cref*{#1}}.\fi
  \par
}
\newcommand{\DeferredProofAnchor}[1]{\refstepcounter{deferredproofanchor}\label{#1}}
\ifusedeferredproofs
  \NewDocumentEnvironment{deferredtheorem}{o +b}{\IfNoValueTF{#1}
      {\begin{theoremE}[][proof at the end, restate, category=deferred]}{\begin{theoremE}[#1][proof at the end, restate, category=deferred]}#2\end{theoremE}}{}
  \NewDocumentEnvironment{deferredlemma}{o +b}{\IfNoValueTF{#1}
      {\begin{lemmaE}[][proof at the end, restate, category=deferred]}{\begin{lemmaE}[#1][proof at the end, restate, category=deferred]}#2\end{lemmaE}}{}
  \NewDocumentEnvironment{deferredproposition}{o +b}{\IfNoValueTF{#1}
      {\begin{propositionE}[][proof at the end, restate, category=deferred]}{\begin{propositionE}[#1][proof at the end, restate, category=deferred]}#2\end{propositionE}}{}
  \NewDocumentEnvironment{deferredproof}{o +b}{\IfNoValueF{#1}{\DeferredProofNotice{#1}}\IfNoValueTF{#1}{\begin{textAtEnd}[category=deferred]
\begin{proof}
#2
\end{proof}
      \end{textAtEnd}
    }{\eraseIfNeeded{deferred}\appendtofile{\pratendGeneratePrefixFile{\jobname}deferred.tex}{\string\makeatletter\string\Hy@SaveLastskip \string\phantomsection\string\label{#1}\string\Hy@RestoreLastskip\string\makeatother \string\begin{proof}\detokenize{#2}\string\end{proof}}}}{}
  \NewDocumentEnvironment{deferredproofsubsection}{m m +b}{\DeferredProofNotice{#1}\begin{textAtEnd}[category=deferredproofsections]
\subsection{#2}
\label{#1}
\begin{proof}
#3
\end{proof}
    \end{textAtEnd}
  }{}
  \newcommand{\PrintDeferredProofs}{\IfFileExists{\pratendGeneratePrefixFile{\jobname}deferred.tex}{\printProofs[deferred]}{}}
  \newcommand{\PrintDeferredProofSections}{\IfFileExists{\pratendGeneratePrefixFile{\jobname}deferredproofsections.tex}{\printProofs[deferredproofsections]}{}}
\else
  \NewDocumentEnvironment{deferredtheorem}{o +b}{\IfNoValueTF{#1}{\begin{theorem}}{\begin{theorem}[#1]}#2\end{theorem}}{}
  \NewDocumentEnvironment{deferredlemma}{o +b}{\IfNoValueTF{#1}{\begin{lemma}}{\begin{lemma}[#1]}#2\end{lemma}}{}
  \NewDocumentEnvironment{deferredproposition}{o +b}{\IfNoValueTF{#1}{\begin{proposition}}{\begin{proposition}[#1]}#2\end{proposition}}{}
  \NewDocumentEnvironment{deferredproof}{o +b}{\IfNoValueTF{#1}{\begin{proof}
        #2\end{proof}
    }{\begin{proof}
        \DeferredProofAnchor{#1}#2\end{proof}
    }}{}
  \NewDocumentEnvironment{deferredproofsubsection}{m m +b}{\subsection{#2}\label{#1}
    \begin{proof}
      #3\end{proof}
  }{}
  \newcommand{\PrintDeferredProofs}{}
  \newcommand{\PrintDeferredProofSections}{}
\fi

\makeatletter
\def\mathcolor#1#{\@mathcolor{#1}}
\def\@mathcolor#1#2#3{\protect\leavevmode
  \begingroup
    \color#1{#2}#3\endgroup
}
\makeatother

\makeatletter
\newenvironment{proof*}[1][\proofname]{\par
  \pushQED{\qed}\normalfont \partopsep=\z@skip \topsep=\z@skip
  \trivlist
  \item[\hskip\labelsep
        \itshape
    #1\@addpunct{.}]\ignorespaces
}{\popQED\endtrivlist\@endpefalse
}
\makeatother

\makeatletter
\AtBeginDocument{

  \@ifpackageloaded{cleveref}{\crefname{theorem}{Theorem}{Theorems}
    \Crefname{theorem}{Theorem}{Theorems}
    \crefname{corollary}{Corollary}{Corollaries}
    \Crefname{corollary}{Corollary}{Corollaries}
    \crefname{conjecture}{Conjecture}{Conjectures}
    \Crefname{conjecture}{Conjecture}{Conjectures}
    \crefname{assumption}{Assumption}{Assumptions}
    \Crefname{assumption}{Assumption}{Assumptions}
    \crefname{lemma}{Lemma}{Lemmas}
    \Crefname{lemma}{Lemma}{Lemmas}
    \crefname{proposition}{Proposition}{Propositions}
    \Crefname{proposition}{Proposition}{Propositions}
    \crefname{claim}{Claim}{Claims}
    \Crefname{claim}{Claim}{Claims}
    \crefname{fact}{Fact}{Facts}
    \Crefname{fact}{Fact}{Facts}
    \crefname{definition}{Definition}{Definitions}
    \Crefname{definition}{Definition}{Definitions}
    \crefname{remark}{Remark}{Remarks}
    \Crefname{remark}{Remark}{Remarks}
    \crefname{example}{Example}{Examples}
    \Crefname{example}{Example}{Examples}
    \crefname{observation}{Observation}{Observations}
    \Crefname{observation}{Observation}{Observations}
    \crefname{maintheorem}{Main Theorem}{Main Theorems}
    \Crefname{maintheorem}{Main Theorem}{Main Theorems}
  }{}
}
\makeatother
 \HideAllResearchBoxes
\ShowChanges

\usepackage[numbers,square,comma,sort&compress]{natbib}
\usepackage[margin=1in]{geometry}
\usepackage{charter}
\usepackage{tabularx}
\usepackage{graphicx}

\title{Bounded-Support Additive Latin Transversals 
} \author{Antoine Deza\\
  McMaster University, Hamilton, Ontario, Canada\\
  \href{mailto:deza@mcmaster.ca}{\texttt{deza@mcmaster.ca}}
  \and
  Yan Gerard\\
  Universit\'e Clermont Auvergne, LIMOS, France\\
  \href{mailto:yan.gerard@uca.fr}{\texttt{yan.gerard@uca.fr}}
  \and
  Yijun Ma\\
  McMaster University, Hamilton, Ontario, Canada\\
  \href{mailto:yijun@mcmaster.ca}{\texttt{yijun@mcmaster.ca}}
  \and
  Sebastian Pokutta\\
  Zuse Institute Berlin and Technische Universit\"at Berlin\\
  \href{mailto:pokutta@zib.de}{\texttt{pokutta@zib.de}}
}
\date{\today}
 
\makeatletter\HideClarify
\begin{document}

\maketitle

\begin{abstract}
Let $A=(a_1,\ldots,a_k)\in(\mathbb Z_m)^k$, let $B\subseteq\mathbb Z_m$ have cardinality $k$, and write $s=|\operatorname{supp}(A)|$. The additive-transversal problem seeks an ordering $b_1,\ldots,b_k$ of $B$ such that the sums $a_i+b_i$ are pairwise distinct. We reduce this problem to exact-weight perfect matching. The support values of $A$ define edge colors in a bipartite graph, and the prescribed multiplicities of all but one color are encoded as a single target weight in base $k+1$. Applying the exact-weight matching algorithm of Lassota, {\L}ukasiewicz, and Polak~\cite[Theorem~2]{lassotaLukasiewiczPolak2022exactMatching} yields a randomized algorithm with running time $(k+1)^{s-1}\operatorname{poly}(k,s,\log m)$.
Thus the problem is solvable in randomized polynomial time for every fixed support size $s$. The reduction applies over arbitrary cyclic groups and does not use the primality assumption appearing in the motivating existence theorem of Alon.
\end{abstract}

\section{Introduction}

In a January 2022 blog post, Eppstein highlighted the following open algorithmic
question from a talk of Alon~\cite{eppstein2022open}. Fix a prime $p$, a
sequence $A=(a_1,\dots,a_k)\in {(\Z_p)}^k$ with $k<p$, and a $k$-subset
$B\subseteq \Z_p$. Can one order the elements of $B$ as $b_1,\dots,b_k$ so
that the sums $a_i+b_i$ are pairwise distinct? Alon proved that such an
ordering always exists~\cite{alon2000additive}, but no polynomial-time construction is known.
Consider the general case of an arbitrary cyclic group $\Z _m$. Over $\mathbb R$, pairwise distinct sums can be obtained by ordering the
elements of $A$ and $B$ increasingly.
The full-size case where $k=m$ is governed by Hall's theorem~\cite{Hall1952}\footnote{Hall's theorem appears in the mathematics of juggling: in the language of periodic juggling patterns, it provides a necessary and sufficient condition for the feasibility of a prescribed throw sequence.}: For a sequence $A=(a_1,\dots,a_m)\in(\mathbb Z_m)^m$, there is an ordering $b_1,\dots,b_m$ of all elements of $\mathbb Z_m$ such that the sums $a_i+b_i$ are pairwise distinct if and only if $\sum_{i=1}^m a_i\equiv 0 \pmod m$.
Hall's theorem extends to the case when $A$ and $B$ have size $k=m-1$. Hence the cases $k=m$ and $k=m-1$ are completely characterized with the important remark that the solutions are constructible in polynomial time~\cite{Hall1952}.
Apart from the brute-force polynomiality for fixed $k$, the cases $k=m$ and $k=m-1$ are the only unrestricted regimes for which a polynomial-time construction is known.

Let $\supp(A)=\{\alpha_0,\dots,\alpha_{s-1}\}$ be the set of distinct values appearing in $A$. Our principal contribution is a direct reduction of bounded-support additive transversals to exact-weight perfect matching. Noticing that all but one of the support multiplicities can be encoded as digits of a single target weight in base $k+1$, and applying the exact-weight matching algorithm of Lassota, {\L}ukasiewicz, and Polak~\cite[Theorem~2]{lassotaLukasiewiczPolak2022exactMatching} gives randomized search time
$(k+1)^{s-1}\operatorname{poly}(k,s,\log m)$. Thus, additive transversals are randomized polynomial-time constructible for every fixed support size.
We use \emph{Color-Counted Matching} as an interface for the reduction: given an edge-colored graph and target counts $r_0,\ldots,r_{h-1}$, find a matching containing exactly $r_i$ edges of color $i$. The additive-transversal problem has a small, direct matching representation to which the known exact-weight algorithm applies. A more self-contained route through Exact Red Matching and the algorithm of Mulmuley, Vazirani, and Vazirani~\cite{mulmuley1987matching} is also presented.

\paragraph{Organization.}
Section~\ref{sec:related} recalls the relevant background on additive
transversals and exact matching. Section~\ref{sec:main-results} states the
main results on additive transversals and Color-Counted Matching, together
with several applications. Section~\ref{sec:proof-bounded-support} proves
Theorem~\ref{thm:bounded-support}. Finally, Section~\ref{sec:self-contained-ccm} establishes
Theorem~\ref{thm:ccm}, a self-contained alternative to the exact-weight
Color-Counted Matching bound.

\section{Related Results}\label{sec:related}

Latin squares and their transversals have a long history in combinatorics. In
the additive setting over a cyclic group, Hall's theorem~\cite{Hall1952} gives a complete characterization in the
full-size case. Namely, for a sequence
$A=(a_1,\dots,a_m)\in(\Z_m)^m$, there exists an ordering
$b_1,\dots,b_m$ of all elements of $\Z_m$ such that the sums $a_i+b_i$ are
pairwise distinct if and only if $\sum_{i=1}^m a_i \equiv 0 \pmod m$.
Moreover, Hall's theorem yields a polynomial-time construction. The near-full
case $k=m-1$ also follows from Hall's theorem and is likewise polynomial-time
constructible.

Alon~\cite{alon2000additive} proved in 2000 that, for prime $p$, every sequence
$a_1,\dots,a_k\in \Z_p$ with $k<p$ and every $k$-subset $B\subseteq \Z_p$
admits an ordering $b_1,\dots,b_k$ such that the sums $a_i+b_i$ are pairwise
distinct. His proof uses the Combinatorial Nullstellensatz applied to the
polynomial
\[
 \prod_{1\le i<j\le k}(x_i-x_j)(a_i+x_i-a_j-x_j).
\]
Dasgupta, K{\'a}rolyi, Serra, and Szegedy extended this existence phenomenon
to $\Z_{p^\alpha}$ and to $(\Z_p)^\alpha$
\cite{dasgupta2001transversals}. Gao and Wang later gave an alternative proof
using group rings~\cite{gao2004group}. These results provide a strong
existence theory, but they do not give a polynomial-time construction for the
general repeated-sequence problem. Apart from the brute-force polynomiality
when $k=O(1)$ and the full or near-full cases $k=m$ and $k=m-1$, no general
polynomial-time construction is known.

Alon's and Hall's additive-transversal theorems are naturally connected with discrete tomography.
In discrete tomography one studies inverse problems in which a finite set of grid points, or equivalently a binary matrix, is reconstructed from prescribed line sums in several directions; these line sums are usually called X-rays~\cite{HermanKuba1999,GardnerGritzmann1997,GardnerGritzmannPrangenberg1999}.

From this viewpoint, Hall's theorem on abelian groups~\cite{Hall1952}, Del Lungo's reconstruction theorem for sums of permutation matrices from diagonal sums~\cite{DelLungo2002}, and Alon's theorem can all be interpreted as periodic X-ray consistency results for permutation-type incidence structures.
The chosen pairs form a periodic binary incidence matrix: one projection
prescribes the selected elements of $B$, another projection prescribes the
multiset $A$, and the all-different condition asks that the third projection
be binary. This geometric viewpoint, which links our reconstruction problem to
discrete tomography, is illustrated in
Fig.~\ref{fig:tomo-colored-matching}.
In this illustration, the multiset is
$A=(0,2,2,3,5)$ and the set is $B=\{0,1,4,5,6\}$ in $\mathbb Z_7$.
Rows are indexed by elements of $\mathbb Z_7$ and encode the choice of
elements of $B$. A row indexed by $b\in B$ must contain one selected cell, and
a row indexed by $b\notin B$ must contain none. Columns encode the
multiplicities in $A$. The column indexed by $a$ must contain as many selected cells as the
multiplicity of $a$ in $A$. Thus, in the example, the column indexed by $0$
contains one selected cell, since $0$ has multiplicity $1$ in $A$. The column
indexed by $1$ contains no selected cell, since $1$ has multiplicity $0$ in
$A$ and the column indexed by $2$ contains two selected cells, since $2$ has
multiplicity $2$ in $A$.
Finally, cyclic diagonals encode sums modulo $7$. Two selected cells lie on the
same diagonal precisely when the corresponding sums $a+b$ are equal. Hence the
additive-transversal condition is exactly the requirement that the selected
cells occupy distinct  diagonals.

\begin{figure}[htb]
 \centering
 \resizebox{\textwidth}{!}{

\definecolor{tomoPink}{RGB}{255,205,205}
\definecolor{tomoPinkLight}{RGB}{255,235,235}
\definecolor{tomoYellow}{RGB}{255,192,0}
\definecolor{tomoYellowLight}{RGB}{255,238,165}
\definecolor{tomoGreen}{RGB}{146,208,80}
\definecolor{tomoGreenLight}{RGB}{231,245,217}
\definecolor{tomoLavender}{RGB}{220,197,237}
\definecolor{tomoLavenderLight}{RGB}{246,239,251}
\definecolor{tomoSlate}{RGB}{173,185,202}
\definecolor{tomoSlateLight}{RGB}{222,228,236}
\definecolor{tomoRed}{RGB}{255,151,151}
\definecolor{tomoRedLight}{RGB}{255,215,215}
\definecolor{tomoSideRed}{RGB}{255,118,118}
\definecolor{tomoSideRedLight}{RGB}{255,190,190}
\definecolor{tomoBlue}{RGB}{68,114,196}
\definecolor{tomoBlueLight}{RGB}{170,195,235}
\definecolor{tomoSelector}{RGB}{0,0,0}
\definecolor{tomoSelectorLight}{RGB}{95,95,95}
\definecolor{tomoLightGray}{RGB}{192,192,192}
\definecolor{tomoLightGrayLight}{RGB}{245,245,245}
\definecolor{tomoPanelGray}{RGB}{166,166,166}
\definecolor{tomoGridGray}{RGB}{166,166,166}

\begin{tikzpicture}[x=1pt,y=-1pt,line cap=round,line join=round]
  \path[use as bounding box] (0,0) rectangle (2057,661);
  \fill[white] (0,0) rectangle (2057,661);
  \tikzset{
    tomoBallNumber/.style={font=\bfseries\sffamily\fontsize{18}{18}\selectfont, text=white, inner sep=0pt},
    tomoGridNumber/.style={font=\bfseries\sffamily\fontsize{18}{18}\selectfont, text=black, inner sep=0pt}
  }

\fill[tomoPanelGray] (144,231) ellipse[x radius=94,y radius=117];
  \fill[tomoPanelGray] (146,488) ellipse[x radius=94,y radius=117];
\node[font=\bfseries\sffamily\fontsize{46}{46}\selectfont, text=tomoBlue, inner sep=0pt] at (52,116) {A};
  \node[font=\bfseries\sffamily\fontsize{46}{46}\selectfont, text=tomoSideRed, inner sep=0pt] at (35,403) {B};

\shade[inner color=tomoPinkLight, outer color=tomoPink] (484,143) circle[radius=24.5];
  \shade[inner color=tomoYellowLight, outer color=tomoYellow] (549,145) circle[radius=24.5];
  \shade[inner color=tomoGreenLight, outer color=tomoGreen] (615,145) circle[radius=24.5];
  \shade[inner color=tomoSlateLight, outer color=tomoSlate] (743,145) circle[radius=24.5];
  \shade[inner color=tomoLavenderLight, outer color=tomoLavender] (677,146) circle[radius=24.5];
  \shade[inner color=tomoRedLight, outer color=tomoRed] (807,146) circle[radius=24.5];
  \shade[inner color=tomoLightGrayLight, outer color=tomoLightGray] (872,146) circle[radius=24.5];
  \shade[inner color=tomoYellowLight, outer color=tomoYellow] (483,209) circle[radius=24.5];
  \shade[inner color=tomoGreenLight, outer color=tomoGreen] (549,211) circle[radius=24.5];
  \shade[inner color=tomoLavenderLight, outer color=tomoLavender] (614,211) circle[radius=24.5];
  \shade[inner color=tomoRedLight, outer color=tomoRed] (742,211) circle[radius=24.5];
  \shade[inner color=tomoLightGrayLight, outer color=tomoLightGray] (806,211) circle[radius=24.5];
  \shade[inner color=tomoSlateLight, outer color=tomoSlate] (676,212) circle[radius=24.5];
  \shade[inner color=tomoPinkLight, outer color=tomoPink] (871,213) circle[radius=24.5];
  \shade[inner color=tomoGreenLight, outer color=tomoGreen] (483,278) circle[radius=24.5];
  \shade[inner color=tomoLavenderLight, outer color=tomoLavender] (548,279) circle[radius=24.5];
  \shade[inner color=tomoSlateLight, outer color=tomoSlate] (613,279) circle[radius=24.5];
  \shade[inner color=tomoLightGrayLight, outer color=tomoLightGray] (741,279) circle[radius=24.5];
  \shade[inner color=tomoYellowLight, outer color=tomoYellow] (868,279) circle[radius=24.5];
  \shade[inner color=tomoRedLight, outer color=tomoRed] (675,280) circle[radius=24.5];
  \shade[inner color=tomoPinkLight, outer color=tomoPink] (806,280) circle[radius=24.5];
  \shade[inner color=tomoGreenLight, outer color=tomoGreen] (869,345) circle[radius=24.5];
  \shade[inner color=tomoLavenderLight, outer color=tomoLavender] (483,347) circle[radius=24.5];
  \shade[inner color=tomoRedLight, outer color=tomoRed] (614,347) circle[radius=24.5];
  \shade[inner color=tomoPinkLight, outer color=tomoPink] (741,347) circle[radius=24.5];
  \shade[inner color=tomoSlateLight, outer color=tomoSlate] (549,349) circle[radius=24.5];
  \shade[inner color=tomoLightGrayLight, outer color=tomoLightGray] (676,349) circle[radius=24.5];
  \shade[inner color=tomoYellowLight, outer color=tomoYellow] (805,349) circle[radius=24.5];
  \shade[inner color=tomoYellowLight, outer color=tomoYellow] (741,415) circle[radius=24.5];
  \shade[inner color=tomoGreenLight, outer color=tomoGreen] (805,415) circle[radius=24.5];
  \shade[inner color=tomoSlateLight, outer color=tomoSlate] (483,416) circle[radius=24.5];
  \shade[inner color=tomoPinkLight, outer color=tomoPink] (677,416) circle[radius=24.5];
  \shade[inner color=tomoLavenderLight, outer color=tomoLavender] (868,416) circle[radius=24.5];
  \shade[inner color=tomoRedLight, outer color=tomoRed] (548,417) circle[radius=24.5];
  \shade[inner color=tomoLightGrayLight, outer color=tomoLightGray] (613,417) circle[radius=24.5];
  \shade[inner color=tomoYellowLight, outer color=tomoYellow] (675,478) circle[radius=24.5];
  \shade[inner color=tomoSlateLight, outer color=tomoSlate] (871,478) circle[radius=24.5];
  \shade[inner color=tomoRedLight, outer color=tomoRed] (483,479) circle[radius=24.5];
  \shade[inner color=tomoLavenderLight, outer color=tomoLavender] (805,479) circle[radius=24.5];
  \shade[inner color=tomoLightGrayLight, outer color=tomoLightGray] (548,481) circle[radius=24.5];
  \shade[inner color=tomoGreenLight, outer color=tomoGreen] (741,481) circle[radius=24.5];
  \shade[inner color=tomoPinkLight, outer color=tomoPink] (613,482) circle[radius=24.5];
  \shade[inner color=tomoGreenLight, outer color=tomoGreen] (675,543) circle[radius=24.5];
  \shade[inner color=tomoLavenderLight, outer color=tomoLavender] (739,543) circle[radius=24.5];
  \shade[inner color=tomoSlateLight, outer color=tomoSlate] (805,544) circle[radius=24.5];
  \shade[inner color=tomoRedLight, outer color=tomoRed] (869,544) circle[radius=24.5];
  \shade[inner color=tomoLightGrayLight, outer color=tomoLightGray] (482,545) circle[radius=24.5];
  \shade[inner color=tomoYellowLight, outer color=tomoYellow] (613,545) circle[radius=24.5];
  \shade[inner color=tomoPinkLight, outer color=tomoPink] (548,547) circle[radius=24.5];
  \shade[inner color=tomoBlueLight, outer color=tomoBlue] (110,172) circle[radius=24.5];
  \shade[inner color=tomoBlueLight, outer color=tomoBlue] (172,185) circle[radius=24.5];
  \shade[inner color=tomoBlueLight, outer color=tomoBlue] (99,251) circle[radius=24.5];
  \shade[inner color=tomoBlueLight, outer color=tomoBlue] (172,260) circle[radius=24.5];
  \shade[inner color=tomoBlueLight, outer color=tomoBlue] (134,309) circle[radius=24.5];
  \shade[inner color=tomoSideRedLight, outer color=tomoSideRed] (114,429) circle[radius=24.5];
  \shade[inner color=tomoSideRedLight, outer color=tomoSideRed] (177,442) circle[radius=24.5];
  \shade[inner color=tomoSideRedLight, outer color=tomoSideRed] (102,508) circle[radius=24.5];
  \shade[inner color=tomoSideRedLight, outer color=tomoSideRed] (177,515) circle[radius=24.5];
  \shade[inner color=tomoSideRedLight, outer color=tomoSideRed] (139,565) circle[radius=24.5];
  \shade[inner color=tomoSelectorLight, outer color=tomoSelector] (1161,143) circle[radius=24.5];
  \shade[inner color=tomoYellowLight, outer color=tomoYellow] (1227,145) circle[radius=24.5];
  \shade[inner color=tomoGreenLight, outer color=tomoGreen] (1292,145) circle[radius=24.5];
  \shade[inner color=tomoSlateLight, outer color=tomoSlate] (1420,145) circle[radius=24.5];
  \shade[inner color=tomoRedLight, outer color=tomoRed] (1484,145) circle[radius=24.5];
  \shade[inner color=tomoLightGrayLight, outer color=tomoLightGray] (1549,145) circle[radius=24.5];
  \shade[inner color=tomoLavenderLight, outer color=tomoLavender] (1354,146) circle[radius=24.5];
  \shade[inner color=tomoYellowLight, outer color=tomoYellow] (1161,208) circle[radius=24.5];
  \shade[inner color=tomoGreenLight, outer color=tomoGreen] (1226,209) circle[radius=24.5];
  \shade[inner color=tomoSelectorLight, outer color=tomoSelector] (1291,209) circle[radius=24.5];
  \shade[inner color=tomoRedLight, outer color=tomoRed] (1419,209) circle[radius=24.5];
  \shade[inner color=tomoSlateLight, outer color=tomoSlate] (1353,211) circle[radius=24.5];
  \shade[inner color=tomoLightGrayLight, outer color=tomoLightGray] (1483,211) circle[radius=24.5];
  \shade[inner color=tomoPinkLight, outer color=tomoPink] (1547,212) circle[radius=24.5];
  \shade[inner color=tomoGreenLight, outer color=tomoGreen] (1160,278) circle[radius=24.5];
  \shade[inner color=tomoLavenderLight, outer color=tomoLavender] (1225,278) circle[radius=24.5];
  \shade[inner color=tomoLightGrayLight, outer color=tomoLightGray] (1418,278) circle[radius=24.5];
  \shade[inner color=tomoYellowLight, outer color=tomoYellow] (1545,278) circle[radius=24.5];
  \shade[inner color=tomoSlateLight, outer color=tomoSlate] (1289,279) circle[radius=24.5];
  \shade[inner color=tomoRedLight, outer color=tomoRed] (1353,279) circle[radius=24.5];
  \shade[inner color=tomoPinkLight, outer color=tomoPink] (1483,280) circle[radius=24.5];
  \shade[inner color=tomoGreenLight, outer color=tomoGreen] (1546,345) circle[radius=24.5];
  \shade[inner color=tomoLavenderLight, outer color=tomoLavender] (1160,346) circle[radius=24.5];
  \shade[inner color=tomoPinkLight, outer color=tomoPink] (1418,346) circle[radius=24.5];
  \shade[inner color=tomoSlateLight, outer color=tomoSlate] (1226,347) circle[radius=24.5];
  \shade[inner color=tomoRedLight, outer color=tomoRed] (1291,347) circle[radius=24.5];
  \shade[inner color=tomoLightGrayLight, outer color=tomoLightGray] (1353,349) circle[radius=24.5];
  \shade[inner color=tomoYellowLight, outer color=tomoYellow] (1483,349) circle[radius=24.5];
  \shade[inner color=tomoGreenLight, outer color=tomoGreen] (1481,413) circle[radius=24.5];
  \shade[inner color=tomoSlateLight, outer color=tomoSlate] (1160,415) circle[radius=24.5];
  \shade[inner color=tomoPinkLight, outer color=tomoPink] (1353,415) circle[radius=24.5];
  \shade[inner color=tomoYellowLight, outer color=tomoYellow] (1418,415) circle[radius=24.5];
  \shade[inner color=tomoRedLight, outer color=tomoRed] (1225,416) circle[radius=24.5];
  \shade[inner color=tomoSelectorLight, outer color=tomoSelector] (1289,416) circle[radius=24.5];
  \shade[inner color=tomoLavenderLight, outer color=tomoLavender] (1545,416) circle[radius=24.5];
  \shade[inner color=tomoSelectorLight, outer color=tomoSelector] (1351,477) circle[radius=24.5];
  \shade[inner color=tomoSlateLight, outer color=tomoSlate] (1547,477) circle[radius=24.5];
  \shade[inner color=tomoRedLight, outer color=tomoRed] (1160,478) circle[radius=24.5];
  \shade[inner color=tomoLavenderLight, outer color=tomoLavender] (1481,478) circle[radius=24.5];
  \shade[inner color=tomoLightGrayLight, outer color=tomoLightGray] (1225,479) circle[radius=24.5];
  \shade[inner color=tomoGreenLight, outer color=tomoGreen] (1418,479) circle[radius=24.5];
  \shade[inner color=tomoPinkLight, outer color=tomoPink] (1289,481) circle[radius=24.5];
  \shade[inner color=tomoGreenLight, outer color=tomoGreen] (1352,543) circle[radius=24.5];
  \shade[inner color=tomoLavenderLight, outer color=tomoLavender] (1415,543) circle[radius=24.5];
  \shade[inner color=tomoSelectorLight, outer color=tomoSelector] (1481,543) circle[radius=24.5];
  \shade[inner color=tomoRedLight, outer color=tomoRed] (1545,543) circle[radius=24.5];
  \shade[inner color=tomoLightGrayLight, outer color=tomoLightGray] (1159,544) circle[radius=24.5];
  \shade[inner color=tomoYellowLight, outer color=tomoYellow] (1289,544) circle[radius=24.5];
  \shade[inner color=tomoPinkLight, outer color=tomoPink] (1225,547) circle[radius=24.5];
  \shade[inner color=tomoBlueLight, outer color=tomoBlue] (1825,271) circle[radius=24.5];
  \shade[inner color=tomoBlueLight, outer color=tomoBlue] (1887,271) circle[radius=24.5];
  \shade[inner color=tomoBlueLight, outer color=tomoBlue] (1951,271) circle[radius=24.5];
  \shade[inner color=tomoBlueLight, outer color=tomoBlue] (2014,271) circle[radius=24.5];
  \shade[inner color=tomoBlueLight, outer color=tomoBlue] (1756,272) circle[radius=24.5];
  \shade[inner color=tomoSideRedLight, outer color=tomoSideRed] (1825,337) circle[radius=24.5];
  \shade[inner color=tomoSideRedLight, outer color=tomoSideRed] (1952,337) circle[radius=24.5];
  \shade[inner color=tomoSideRedLight, outer color=tomoSideRed] (1888,338) circle[radius=24.5];
  \shade[inner color=tomoSideRedLight, outer color=tomoSideRed] (2017,338) circle[radius=24.5];
  \shade[inner color=tomoSideRedLight, outer color=tomoSideRed] (1757,339) circle[radius=24.5];
  \shade[inner color=tomoSelectorLight, outer color=tomoSelector] (1952,434) circle[radius=24.5];
  \shade[inner color=tomoSelectorLight, outer color=tomoSelector] (1756,435) circle[radius=24.5];
  \shade[inner color=tomoSelectorLight, outer color=tomoSelector] (1887,435) circle[radius=24.5];
  \shade[inner color=tomoSelectorLight, outer color=tomoSelector] (2018,435) circle[radius=24.5];
  \shade[inner color=tomoSelectorLight, outer color=tomoSelector] (1826,436) circle[radius=24.5];

\draw[tomoGridGray,line width=4pt] (515.5,110) -- (515.5,579);
  \draw[tomoGridGray,line width=4pt] (580.5,110) -- (580.5,579);
  \draw[tomoGridGray,line width=4pt] (644.5,110) -- (644.5,579);
  \draw[tomoGridGray,line width=4pt] (709.5,110) -- (709.5,579);
  \draw[tomoGridGray,line width=4pt] (773.5,110) -- (773.5,579);
  \draw[tomoGridGray,line width=4pt] (837.5,110) -- (837.5,579);
  \draw[tomoGridGray,line width=4pt] (447,177.5) -- (905,177.5);
  \draw[tomoGridGray,line width=4pt] (447,246.5) -- (905,246.5);
  \draw[tomoGridGray,line width=4pt] (447,313.5) -- (905,313.5);
  \draw[tomoGridGray,line width=4pt] (447,382.5) -- (905,382.5);
  \draw[tomoGridGray,line width=4pt] (447,447.5) -- (905,447.5);
  \draw[tomoGridGray,line width=4pt] (447,511.5) -- (905,511.5);
  \draw[black,line width=4pt] (447,110) rectangle (905,579);
  \draw[tomoGridGray,line width=4pt] (1192.5,109) -- (1192.5,579);
  \draw[tomoGridGray,line width=4pt] (1257.5,109) -- (1257.5,579);
  \draw[tomoGridGray,line width=4pt] (1321.5,109) -- (1321.5,579);
  \draw[tomoGridGray,line width=4pt] (1386.0,109) -- (1386.0,579);
  \draw[tomoGridGray,line width=4pt] (1450.5,109) -- (1450.5,579);
  \draw[tomoGridGray,line width=4pt] (1514.5,109) -- (1514.5,579);
  \draw[tomoGridGray,line width=4pt] (1124,176.5) -- (1582,176.5);
  \draw[tomoGridGray,line width=4pt] (1124,245.5) -- (1582,245.5);
  \draw[tomoGridGray,line width=4pt] (1124,312.5) -- (1582,312.5);
  \draw[tomoGridGray,line width=4pt] (1124,381.5) -- (1582,381.5);
  \draw[tomoGridGray,line width=4pt] (1124,447.5) -- (1582,447.5);
  \draw[tomoGridGray,line width=4pt] (1124,511.5) -- (1582,511.5);
  \draw[black,line width=4pt] (1124,109) rectangle (1582,579);

\draw[black,line width=5pt] (1725,388) -- (2038,388);

\path[fill=tomoBlue] (480,589) -- (489,600) -- (484,600) -- (484,612) -- (476,612) -- (476,600) -- (471,600) -- cycle;
  \node[font=\bfseries\sffamily\fontsize{17}{17}\selectfont, text=tomoBlue, inner sep=0pt] at (480,634) {1};
  \path[fill=tomoGridGray] (546,589) -- (555,600) -- (550,600) -- (550,612) -- (542,612) -- (542,600) -- (537,600) -- cycle;
  \node[font=\bfseries\sffamily\fontsize{17}{17}\selectfont, text=black, inner sep=0pt] at (546,634) {0};
  \path[fill=tomoBlue] (612,589) -- (621,600) -- (616,600) -- (616,612) -- (608,612) -- (608,600) -- (603,600) -- cycle;
  \node[font=\bfseries\sffamily\fontsize{17}{17}\selectfont, text=tomoBlue, inner sep=0pt] at (612,634) {2};
  \path[fill=tomoBlue] (676,589) -- (685,600) -- (680,600) -- (680,612) -- (672,612) -- (672,600) -- (667,600) -- cycle;
  \node[font=\bfseries\sffamily\fontsize{17}{17}\selectfont, text=tomoBlue, inner sep=0pt] at (676,634) {1};
  \path[fill=tomoGridGray] (740,589) -- (749,600) -- (744,600) -- (744,612) -- (736,612) -- (736,600) -- (731,600) -- cycle;
  \node[font=\bfseries\sffamily\fontsize{17}{17}\selectfont, text=black, inner sep=0pt] at (740,634) {0};
  \path[fill=tomoBlue] (804,589) -- (813,600) -- (808,600) -- (808,612) -- (800,612) -- (800,600) -- (795,600) -- cycle;
  \node[font=\bfseries\sffamily\fontsize{17}{17}\selectfont, text=tomoBlue, inner sep=0pt] at (804,634) {1};
  \path[fill=tomoGridGray] (870,589) -- (879,600) -- (874,600) -- (874,612) -- (866,612) -- (866,600) -- (861,600) -- cycle;
  \node[font=\bfseries\sffamily\fontsize{17}{17}\selectfont, text=black, inner sep=0pt] at (870,634) {0};
  \path[fill=tomoBlue] (1156,589) -- (1165,600) -- (1160,600) -- (1160,612) -- (1152,612) -- (1152,600) -- (1147,600) -- cycle;
  \node[font=\bfseries\sffamily\fontsize{17}{17}\selectfont, text=tomoBlue, inner sep=0pt] at (1156,634) {1};
  \path[fill=tomoGridGray] (1222,589) -- (1231,600) -- (1226,600) -- (1226,612) -- (1218,612) -- (1218,600) -- (1213,600) -- cycle;
  \node[font=\bfseries\sffamily\fontsize{17}{17}\selectfont, text=black, inner sep=0pt] at (1222,634) {0};
  \path[fill=tomoBlue] (1288,589) -- (1297,600) -- (1292,600) -- (1292,612) -- (1284,612) -- (1284,600) -- (1279,600) -- cycle;
  \node[font=\bfseries\sffamily\fontsize{17}{17}\selectfont, text=tomoBlue, inner sep=0pt] at (1288,634) {2};
  \path[fill=tomoBlue] (1354,589) -- (1363,600) -- (1358,600) -- (1358,612) -- (1350,612) -- (1350,600) -- (1345,600) -- cycle;
  \node[font=\bfseries\sffamily\fontsize{17}{17}\selectfont, text=tomoBlue, inner sep=0pt] at (1354,634) {1};
  \path[fill=tomoGridGray] (1417,589) -- (1426,600) -- (1421,600) -- (1421,612) -- (1413,612) -- (1413,600) -- (1408,600) -- cycle;
  \node[font=\bfseries\sffamily\fontsize{17}{17}\selectfont, text=black, inner sep=0pt] at (1417,634) {0};
  \path[fill=tomoBlue] (1482,589) -- (1491,600) -- (1486,600) -- (1486,612) -- (1478,612) -- (1478,600) -- (1473,600) -- cycle;
  \node[font=\bfseries\sffamily\fontsize{17}{17}\selectfont, text=tomoBlue, inner sep=0pt] at (1482,634) {1};
  \path[fill=tomoGridGray] (1547,589) -- (1556,600) -- (1551,600) -- (1551,612) -- (1543,612) -- (1543,600) -- (1538,600) -- cycle;
  \node[font=\bfseries\sffamily\fontsize{17}{17}\selectfont, text=black, inner sep=0pt] at (1547,634) {0};

\node[font=\bfseries\sffamily\fontsize{14}{14}\selectfont, text=tomoSideRed, inner sep=0pt] at (385,155) {1};
  \path[fill=tomoSideRed] (410,149) -- (423,149) -- (423,145) -- (435,154) -- (423,163) -- (423,159) -- (410,159) -- cycle;
  \node[font=\bfseries\sffamily\fontsize{14}{14}\selectfont, text=tomoSideRed, inner sep=0pt] at (385,220) {1};
  \path[fill=tomoSideRed] (410,214) -- (423,214) -- (423,210) -- (435,219) -- (423,228) -- (423,224) -- (410,224) -- cycle;
  \node[font=\bfseries\sffamily\fontsize{14}{14}\selectfont, text=tomoGridGray, inner sep=0pt] at (385,284) {0};
  \path[fill=tomoGridGray] (410,278) -- (423,278) -- (423,274) -- (435,283) -- (423,292) -- (423,288) -- (410,288) -- cycle;
  \node[font=\bfseries\sffamily\fontsize{14}{14}\selectfont, text=tomoGridGray, inner sep=0pt] at (385,349) {0};
  \path[fill=tomoGridGray] (410,343) -- (423,343) -- (423,339) -- (435,348) -- (423,357) -- (423,353) -- (410,353) -- cycle;
  \node[font=\bfseries\sffamily\fontsize{14}{14}\selectfont, text=tomoSideRed, inner sep=0pt] at (385,412) {1};
  \path[fill=tomoSideRed] (410,406) -- (423,406) -- (423,402) -- (435,411) -- (423,420) -- (423,416) -- (410,416) -- cycle;
  \node[font=\bfseries\sffamily\fontsize{14}{14}\selectfont, text=tomoSideRed, inner sep=0pt] at (385,477) {1};
  \path[fill=tomoSideRed] (410,471) -- (423,471) -- (423,467) -- (435,476) -- (423,485) -- (423,481) -- (410,481) -- cycle;
  \node[font=\bfseries\sffamily\fontsize{14}{14}\selectfont, text=tomoSideRed, inner sep=0pt] at (385,545) {1};
  \path[fill=tomoSideRed] (410,539) -- (423,539) -- (423,535) -- (435,544) -- (423,553) -- (423,549) -- (410,549) -- cycle;
  \node[font=\bfseries\sffamily\fontsize{14}{14}\selectfont, text=tomoSideRed, inner sep=0pt] at (1062,155) {1};
  \path[fill=tomoSideRed] (1087,149) -- (1100,149) -- (1100,145) -- (1112,154) -- (1100,163) -- (1100,159) -- (1087,159) -- cycle;
  \node[font=\bfseries\sffamily\fontsize{14}{14}\selectfont, text=tomoSideRed, inner sep=0pt] at (1062,220) {1};
  \path[fill=tomoSideRed] (1087,214) -- (1100,214) -- (1100,210) -- (1112,219) -- (1100,228) -- (1100,224) -- (1087,224) -- cycle;
  \node[font=\bfseries\sffamily\fontsize{14}{14}\selectfont, text=tomoGridGray, inner sep=0pt] at (1062,284) {0};
  \path[fill=tomoGridGray] (1087,278) -- (1100,278) -- (1100,274) -- (1112,283) -- (1100,292) -- (1100,288) -- (1087,288) -- cycle;
  \node[font=\bfseries\sffamily\fontsize{14}{14}\selectfont, text=tomoGridGray, inner sep=0pt] at (1062,349) {0};
  \path[fill=tomoGridGray] (1087,343) -- (1100,343) -- (1100,339) -- (1112,348) -- (1100,357) -- (1100,353) -- (1087,353) -- cycle;
  \node[font=\bfseries\sffamily\fontsize{14}{14}\selectfont, text=tomoSideRed, inner sep=0pt] at (1062,412) {1};
  \path[fill=tomoSideRed] (1087,406) -- (1100,406) -- (1100,402) -- (1112,411) -- (1100,420) -- (1100,416) -- (1087,416) -- cycle;
  \node[font=\bfseries\sffamily\fontsize{14}{14}\selectfont, text=tomoSideRed, inner sep=0pt] at (1062,477) {1};
  \path[fill=tomoSideRed] (1087,471) -- (1100,471) -- (1100,467) -- (1112,476) -- (1100,485) -- (1100,481) -- (1087,481) -- cycle;
  \node[font=\bfseries\sffamily\fontsize{14}{14}\selectfont, text=tomoSideRed, inner sep=0pt] at (1062,545) {1};
  \path[fill=tomoSideRed] (1087,539) -- (1100,539) -- (1100,535) -- (1112,544) -- (1100,553) -- (1100,549) -- (1087,549) -- cycle;

\node[tomoGridNumber] at (484,144) {0};
  \node[tomoBallNumber] at (1161,144) {0};
  \node[tomoGridNumber] at (549,146) {1};
  \node[tomoGridNumber] at (615,146) {2};
  \node[tomoGridNumber] at (743,146) {4};
  \node[tomoGridNumber] at (1227,146) {1};
  \node[tomoGridNumber] at (1292,146) {2};
  \node[tomoGridNumber] at (1420,146) {4};
  \node[tomoGridNumber] at (1484,146) {5};
  \node[tomoGridNumber] at (1549,146) {6};
  \node[tomoGridNumber] at (677,147) {3};
  \node[tomoGridNumber] at (807,147) {5};
  \node[tomoGridNumber] at (872,147) {6};
  \node[tomoGridNumber] at (1354,147) {3};
  \node[tomoGridNumber] at (110,173) {0};
  \node[tomoGridNumber] at (172,186) {2};
  \node[tomoGridNumber] at (1161,209) {1};
  \node[tomoGridNumber] at (483,210) {1};
  \node[tomoGridNumber] at (1226,210) {2};
  \node[tomoBallNumber] at (1291,210) {3};
  \node[tomoGridNumber] at (1419,210) {5};
  \node[tomoGridNumber] at (549,212) {2};
  \node[tomoGridNumber] at (614,212) {3};
  \node[tomoGridNumber] at (742,212) {5};
  \node[tomoGridNumber] at (806,212) {6};
  \node[tomoGridNumber] at (1353,212) {4};
  \node[tomoGridNumber] at (1483,212) {6};
  \node[tomoGridNumber] at (676,213) {4};
  \node[tomoGridNumber] at (1547,213) {0};
  \node[tomoGridNumber] at (871,214) {0};
  \node[tomoGridNumber] at (99,252) {2};
  \node[tomoGridNumber] at (172,261) {3};
  \node[tomoGridNumber] at (1825,272) {2};
  \node[tomoGridNumber] at (1887,272) {2};
  \node[tomoGridNumber] at (1951,272) {3};
  \node[tomoGridNumber] at (2014,272) {5};
  \node[tomoGridNumber] at (1756,273) {0};
  \node[tomoGridNumber] at (483,279) {2};
  \node[tomoGridNumber] at (1160,279) {2};
  \node[tomoGridNumber] at (1225,279) {3};
  \node[tomoGridNumber] at (1418,279) {6};
  \node[tomoGridNumber] at (1545,279) {1};
  \node[tomoGridNumber] at (548,280) {3};
  \node[tomoGridNumber] at (613,280) {4};
  \node[tomoGridNumber] at (741,280) {6};
  \node[tomoGridNumber] at (868,280) {1};
  \node[tomoGridNumber] at (1289,280) {4};
  \node[tomoGridNumber] at (1353,280) {5};
  \node[tomoGridNumber] at (675,281) {5};
  \node[tomoGridNumber] at (806,281) {0};
  \node[tomoGridNumber] at (1483,281) {0};
  \node[tomoGridNumber] at (134,310) {5};
  \node[tomoGridNumber] at (1825,338) {1};
  \node[tomoGridNumber] at (1952,338) {5};
  \node[tomoGridNumber] at (1888,339) {4};
  \node[tomoGridNumber] at (2017,339) {6};
  \node[tomoGridNumber] at (1757,340) {0};
  \node[tomoGridNumber] at (869,346) {2};
  \node[tomoGridNumber] at (1546,346) {2};
  \node[tomoGridNumber] at (1160,347) {3};
  \node[tomoGridNumber] at (1418,347) {0};
  \node[tomoGridNumber] at (483,348) {3};
  \node[tomoGridNumber] at (614,348) {5};
  \node[tomoGridNumber] at (741,348) {0};
  \node[tomoGridNumber] at (1226,348) {4};
  \node[tomoGridNumber] at (1291,348) {5};
  \node[tomoGridNumber] at (549,350) {4};
  \node[tomoGridNumber] at (676,350) {6};
  \node[tomoGridNumber] at (805,350) {1};
  \node[tomoGridNumber] at (1353,350) {6};
  \node[tomoGridNumber] at (1483,350) {1};
  \node[tomoGridNumber] at (1481,414) {2};
  \node[tomoGridNumber] at (741,416) {1};
  \node[tomoGridNumber] at (805,416) {2};
  \node[tomoGridNumber] at (1160,416) {4};
  \node[tomoGridNumber] at (1353,416) {0};
  \node[tomoGridNumber] at (1418,416) {1};
  \node[tomoGridNumber] at (483,417) {4};
  \node[tomoGridNumber] at (677,417) {0};
  \node[tomoGridNumber] at (868,417) {3};
  \node[tomoGridNumber] at (1225,417) {5};
  \node[tomoBallNumber] at (1289,417) {6};
  \node[tomoGridNumber] at (1545,417) {3};
  \node[tomoGridNumber] at (548,418) {5};
  \node[tomoGridNumber] at (613,418) {6};
  \node[tomoGridNumber] at (114,430) {0};
  \node[tomoBallNumber] at (1952,435) {1};
  \node[tomoBallNumber] at (1756,436) {0};
  \node[tomoBallNumber] at (1887,436) {6};
  \node[tomoBallNumber] at (2018,436) {4};
  \node[tomoBallNumber] at (1826,437) {3};
  \node[tomoGridNumber] at (177,443) {1};
  \node[tomoBallNumber] at (1351,478) {1};
  \node[tomoGridNumber] at (1547,478) {4};
  \node[tomoGridNumber] at (675,479) {1};
  \node[tomoGridNumber] at (871,479) {4};
  \node[tomoGridNumber] at (1160,479) {5};
  \node[tomoGridNumber] at (1481,479) {3};
  \node[tomoGridNumber] at (483,480) {5};
  \node[tomoGridNumber] at (805,480) {3};
  \node[tomoGridNumber] at (1225,480) {6};
  \node[tomoGridNumber] at (1418,480) {2};
  \node[tomoGridNumber] at (548,482) {6};
  \node[tomoGridNumber] at (741,482) {2};
  \node[tomoGridNumber] at (1289,482) {0};
  \node[tomoGridNumber] at (613,483) {0};
  \node[tomoGridNumber] at (102,509) {4};
  \node[tomoGridNumber] at (177,516) {5};
  \node[tomoGridNumber] at (675,544) {2};
  \node[tomoGridNumber] at (739,544) {3};
  \node[tomoGridNumber] at (1352,544) {2};
  \node[tomoGridNumber] at (1415,544) {3};
  \node[tomoBallNumber] at (1481,544) {4};
  \node[tomoGridNumber] at (1545,544) {5};
  \node[tomoGridNumber] at (805,545) {4};
  \node[tomoGridNumber] at (869,545) {5};
  \node[tomoGridNumber] at (1159,545) {6};
  \node[tomoGridNumber] at (1289,545) {1};
  \node[tomoGridNumber] at (482,546) {6};
  \node[tomoGridNumber] at (613,546) {1};
  \node[tomoGridNumber] at (548,548) {0};
  \node[tomoGridNumber] at (1225,548) {0};
  \node[tomoGridNumber] at (139,566) {6};

\path[fill=black] (250,335.0) -- (316.3,335.0) -- (316.3,298) -- (357,354.0) -- (316.3,410) -- (316.3,373.0) -- (250,373.0) -- cycle;
  \path[fill=black] (928,331.5) -- (995.0,331.5) -- (995.0,295) -- (1036,350.5) -- (995.0,406) -- (995.0,369.5) -- (928,369.5) -- cycle;
  \path[fill=black] (1608,331.5) -- (1675.0,331.5) -- (1675.0,295) -- (1716,350.5) -- (1675.0,406) -- (1675.0,369.5) -- (1608,369.5) -- cycle;
\end{tikzpicture}
 }
 \caption{An additive Latin transversal instance viewed as a three-direction discrete tomography instance.}
 \label{fig:tomo-colored-matching}
\end{figure}
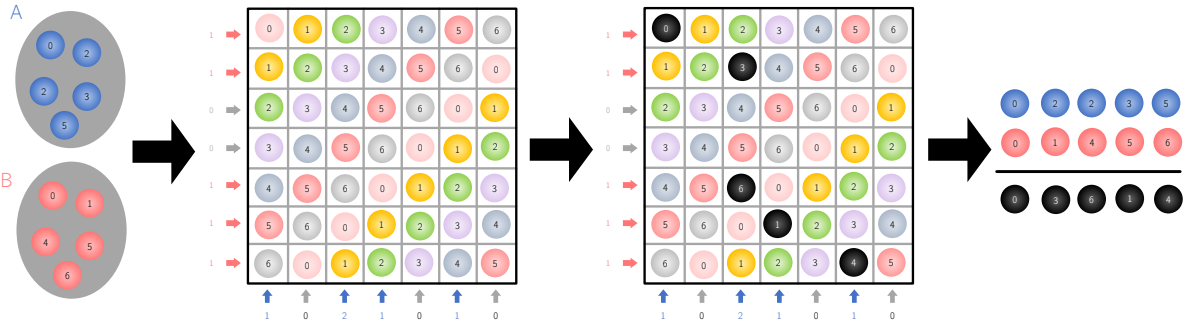

This perspective explains why matching enters the picture.
Thus the bounded-support additive-transversal problem may be viewed as a particularly structured instance of a broader discrete-tomography theme: reconstruction from three projections when one projection has bounded support.

On the algorithmic side, Color-Counted Matching is a primitive that appears
naturally in the additive-transversal problem but is not specific to it. It
includes Exact Matching as the two-color case, and it captures several
exact-composition variants of matching problems. For instance, if all target
counts are $1$, it includes exact rainbow matching. There is also a direct Color-Counted matching formulation of
3-dimensional matching: a triple $(x,y,z)$ may be represented as an edge $xy$
colored by $z$, and requiring prescribed counts for the $z$-colors enforces
the corresponding constraints in the third coordinate. More generally,
Color-Counted Matching captures colored perfect matchings arising in discrete
tomography, tilings, permutation-matrix reconstruction, and cycle-cover
problems. Several applications of this formulation are described below.

The algorithmic bridge used in this paper is the reduction from
Color-Counted Matching to Exact-Weight Perfect Matching. The latter problem
asks, given a non-negatively weighted multigraph $(V,E)$ and a target value $t$, whether $(V,E)$ has a perfect
matching whose weight sum is $t$. Lassota, {\L}ukasiewicz, and Polak~\cite{lassotaLukasiewiczPolak2022exactMatching}
give a randomized Monte Carlo algorithm for this problem running in time $\widetilde O(t|V|^8+|E|)$.
This is the exact-weight matching engine used in our main reduction.
The broader matching-ILP framework of Lassota and
Ligthart~\cite{lassotaLigthart2025matchingIP} applies to matching
problems with additional linear constraints and therefore also
contains prescribed color counts as a special case. Their
framework gives other randomized XP algorithms parameterized by the number
of additional constraints.

\section{Main Results}\label{sec:main-results}

We now state the two algorithmic problems and the
corresponding complexity bounds.

\begin{definition}Fix integers $m\ge1$ and $1\le k\le m$. An Additive Transversal instance consists of a sequence $A=(a_1,\dots,a_k)\in(\Z_m)^k$ and a set $B\subseteq\Z_m$ of cardinality $k$.
A \emph{solution} is an ordering $b_1,\dots,b_k$ of the elements of $B$
such that the sums
$ a_1+b_1,\dots,a_k+b_k $
are pairwise distinct in $\Z_m$.
\end{definition}

Let $\alpha_0,\dots,\alpha_{s-1}$ be the distinct values appearing in $A$, and let
$m_j$ be the multiplicity of $\alpha _j$ in $A$. Thus
$\sum_{j=0}^{s-1} m_j = k.$
Our main result is a randomized algorithm whose running time is
polynomial for every fixed support size $s$, with the support-dependent factor made explicit.

\begin{theorem}
\label{thm:bounded-support}
There is a randomized algorithm such that, on input an additive transversal
instance over $\Z_m$ with support size $s=|\supp(A)|$, the algorithm either
reports failure or outputs a correct solution. If a solution exists, then the
algorithm outputs a solution with probability at least $2/3$ in time
$(k+1)^{s-1}\poly(k,s,\log m)$.
The failure probability can be reduced by independent repetition. In
particular, the problem is randomized polynomial-time solvable for every fixed
support size $s$.
\end{theorem}

The proof of Theorem~\ref{thm:bounded-support}, given in
Section~\ref{sec:proof-bounded-support}, proceeds through a more general
matching primitive, which we call Color-Counted Matching.

\begin{definition}A Color-Counted Matching instance consists of a graph $(V,E)$, an edge coloring $\chi:E\to \{0,\dots,h-1\}$,
and nonnegative integer target counts $r_0,\dots,r_{h-1}$. The task is to find a matching
$M$ such that
$|M\cap \chi^{-1}(i)|=r_i$ for every $i\in \{0,\dots,h-1\}.$
\end{definition}

We denote $q=\sum_{i=0}^{h-1} r_i$. Every feasible solution of a Color-Counted Matching instance has exactly
$q$ edges. Thus, if $|V|=2q$, every feasible solution is automatically a
perfect matching. Similarly, in a bipartite graph $(L\sqcup R,E)$, if
$|L|=q$, then every feasible solution automatically saturates $L$, while
vertices of $R$ may remain unmatched.

We now state the complexity of the Color-Counted Matching algorithm obtained
from the exact-weight Monte Carlo algorithm of Lassota, {\L}ukasiewicz, and
Polak \cite{lassotaLukasiewiczPolak2022exactMatching}. In particular, this gives randomized polynomial time for every fixed
number of colors.

\begin{theorem}
\label{thm:ccm-direct-weighted}
For $q\ge 1$ and $h\ge 2$, feasibility of a Color-Counted Matching instance
can be decided with one-sided error in time $
 \widetilde O\!\left(
 q(q+1)^{h-2}|V|^8+|V|^2+|E|
 \right).$
For every $h\ge 1$, a matching with the prescribed counts, whenever one
exists, can be constructed with success probability at least $2/3$ in time $(q+1)^{h-1}\poly(|V|+|E|).$
\end{theorem}

\begin{proof}
The case $q=0$ is solved by the empty matching, and the case $h=1$ reduces to
finding a matching of size $q$. Assume therefore that $q\ge1$ and $h\ge2$.
If $2q>|V|$, the instance is infeasible. Otherwise add
$|V|-2q$ dummy vertices, join each dummy vertex to every original vertex by a
zero-weight edge, and retain the original edges. Every perfect matching of
the augmented graph uses exactly $q$ original edges, and conversely every
$q$-edge matching of the original graph extends to such a perfect matching.
The augmented graph has at most $2|V|$ vertices and at most
$|E|+|V|^2$ edges.
Set $Q=q+1$. Give every original edge of color $i<h-1$ weight $Q^i$ and
every edge of color $h-1$ weight zero, and set
\[
 T=\sum_{i=0}^{h-2}r_iQ^i.
\]
Equality with $T$ determines the first $h-1$ color counts by uniqueness of
base-$Q$ digits. The omitted count is then determined by the fact that
exactly $q$ original edges are selected. Moreover,
$T\le qQ^{h-2}$. Applying the weighted exact-matching algorithm of Lassota,
{\L}ukasiewicz, and
Polak~\cite[Theorem~2]{lassotaLukasiewiczPolak2022exactMatching} gives the
stated decision bound. Amplifying its one-sided error over the polynomially
many calls in the standard edge-fixing self-reduction constructs a witness
within the stated search bound; any returned matching is checked
deterministically.
\end{proof}

\begin{remark}
Partition a set of $q$ labeled slots into sets $S_i$ with $|S_i|=r_i$.
Replace every edge of color $i$ by one parallel copy for each slot
$j\in S_i$, and assign that copy the weight $2^q+2^{j-1}$. The partition
encoding of Lassota, {\L}ukasiewicz, and
Polak~\cite[Lemma~9]{lassotaLukasiewiczPolak2022exactMatching} shows that an
exact target of $q2^q+(2^q-1)$ forces every slot to be used exactly once.
Together with their Theorem~2 and the same dummy completion, this gives an
alternative $2^q\poly(|V|+|E|)$ randomized search algorithm. Hence one may
use the better of
$(q+1)^{h-1}\poly(|V|+|E|)$ and $2^q\poly(|V|+|E|).$
For the fixed-support application, where $q=k$ can be large and $h=s$ is
fixed, the first bound is the relevant one.
\end{remark}

For completeness, we also give a more self-contained route through Exact Red
Matching and the algorithm of Mulmuley--Vazirani--Vazirani. This route is
conceptually elementary, but it expands weights into unary red-black path
gadgets and therefore gives a weaker parameter dependence than the Monte Carlo algorithm of Lassota--{\L}ukasiewicz--Polak.

\begin{theorem}\label{thm:ccm}
Let $(V,E)$ be a graph whose edges are
colored with colors $0,\dots,h-1$. Let $r_0,\dots,r_{h-1}$ be prescribed target
counts, and set $q=\sum _{i=0}^{h-1} r_i$.
There is a randomized algorithm which either reports failure or outputs a
matching $M$ such that
$
 |M\cap \chi^{-1}(i)|=r_i$
for every $i\in\{0,\dots,h-1\}.
$
If such a matching exists, then the algorithm outputs one with probability at
least $2/3$ in time
$
 \left(|V|^2+|E|(q+1)^{h-1}\right)^{O(1)}.
$
The failure probability can be reduced by independent repetition.
In particular, since any feasible matching has size $q$, the problem is
randomized polynomial-time solvable for every fixed number of colors $h$.
If the input graph is simple, then $|E|\le |V|^2$ and $q\le \frac{|V|}{2}$ on feasible instances,
so the running time is bounded by
$(|V|+|E|)^{O(h)}$.
\end{theorem}

\noindent The algorithms in Theorems~\ref{thm:bounded-support},
\ref{thm:ccm-direct-weighted}, and~\ref{thm:ccm} are randomized because they
rely on randomized exact-matching primitives. Theorem~\ref{thm:ccm-direct-weighted}
uses the exact-weight matching algorithm of Lassota, {\L}ukasiewicz, and
Polak, while the self-contained route leading to Theorem~\ref{thm:ccm} uses
the algorithm of Mulmuley--Vazirani--Vazirani, as explained in
Section~\ref{sec:self-contained-ccm}. Derandomizing these primitives in the
generality needed here is a long-standing open problem.

Beyond additive transversals, Theorems~\ref{thm:ccm-direct-weighted}
and~\ref{thm:ccm} give immediate fixed-color algorithms for several 
problems that can be formulated as Color-Counted Matching instances:

First, the formulation applies to the problem of tiling a region $R$ of the triangular grid with prescribed numbers of lozenges in each orientation
(the problem is illustrated in Fig.~\ref{lozenges}).
Lozenge tilings of finite regions of the triangular lattice are equivalent to
perfect matchings in the dual honeycomb graph. The three possible lozenge
orientations induce a $3$-coloring of the edges. Every tiling corresponds to a perfect matching of
the dual graph with three colors. Either Theorem~\ref{thm:ccm-direct-weighted} or Theorem~\ref{thm:ccm}
therefore gives a randomized algorithm with running time $|R|^{O(1)}$, and
hence polynomial time in the number of triangles of $R$.

\begin{figure}[htb]
\begin{center}
 \includegraphics[width=0.5\textwidth]{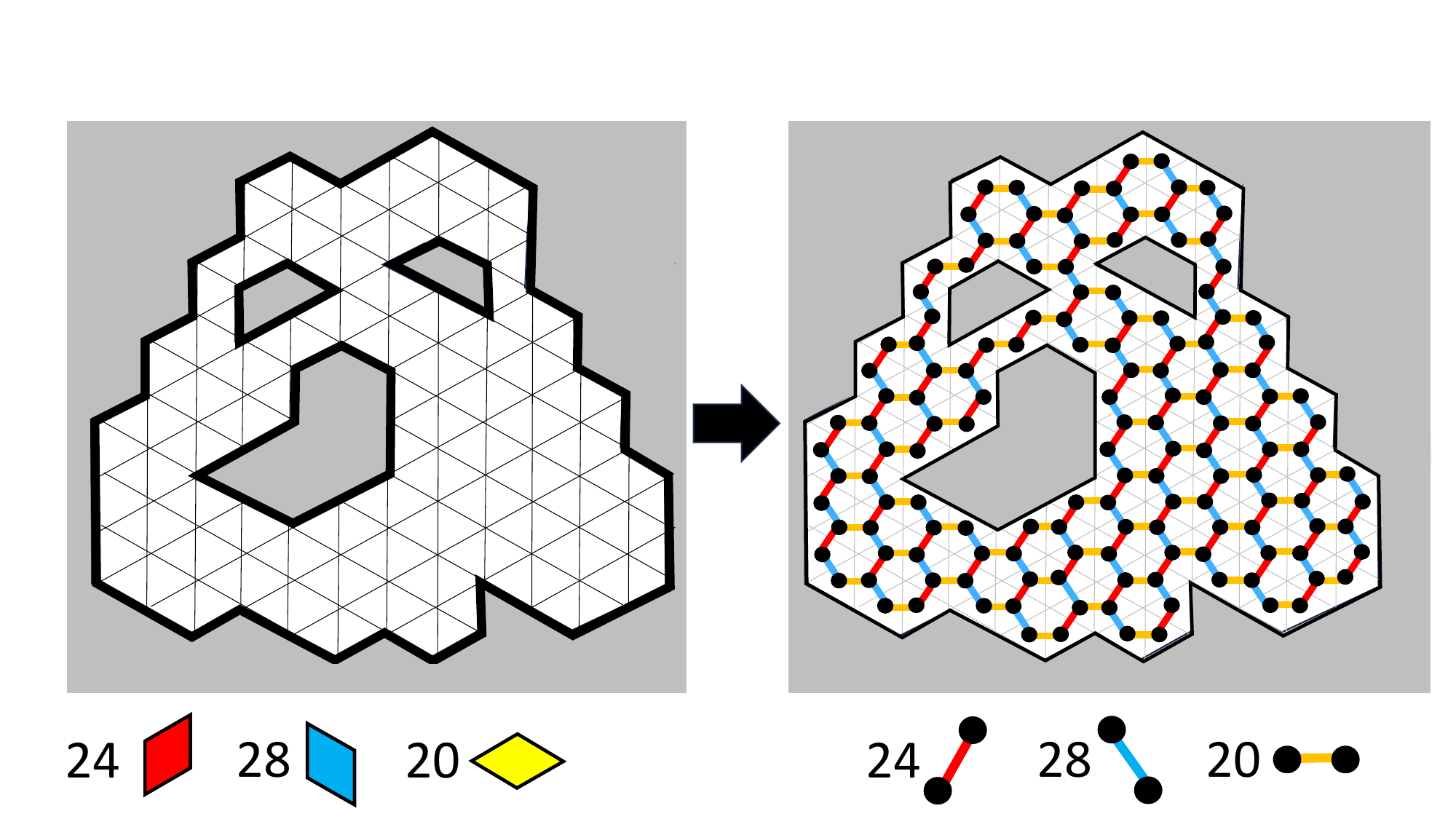}
 \hspace{-0.2cm}
 \includegraphics[width=0.5\textwidth]{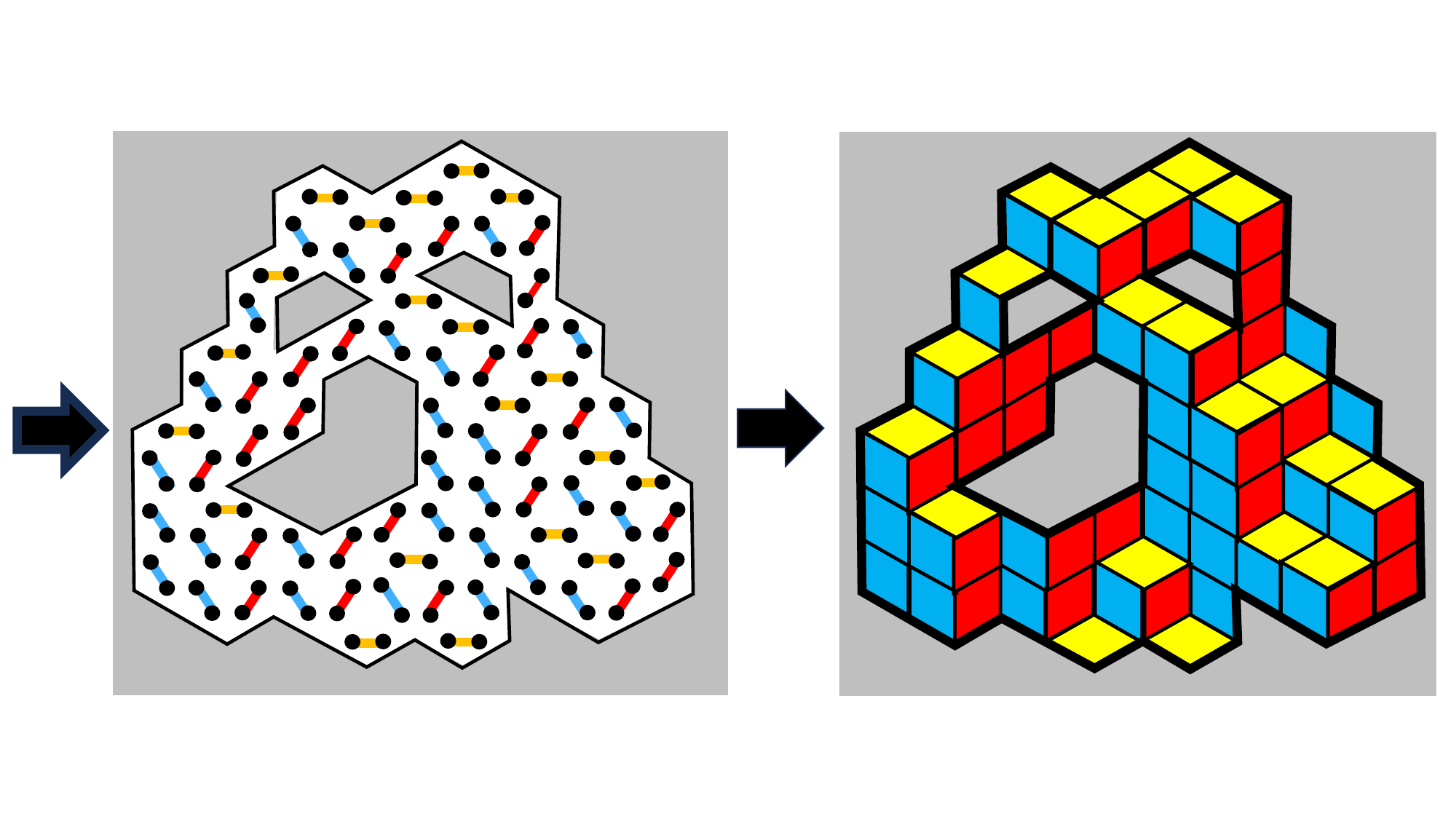}
\end{center}
\caption{\label{lozenges} Reduction of the problem of tiling a region of the triangular grid with prescribed numbers of lozenges in each orientation to a Color-Counted Matching problem with three colors.}
\end{figure}

Second, directed cycle covers with prescribed arc-type counts reduce directly
to Color-Counted Matching. Indeed, a directed cycle cover in a digraph
$D=(V,A)$ is equivalent to a perfect matching in the bipartite graph with left
and right copies of $V$, where an arc $u\to v$ becomes the edge $u_Lv_R$.
If the arcs have $h$ types, then the resulting
Color-Counted Matching instance has $2|V|$ vertices, $|A|$ edges, and target
matching size $|V|$. Thus either Theorem~\ref{thm:ccm-direct-weighted} or Theorem~\ref{thm:ccm} gives a randomized algorithm
with running time $(|V|^2+|A||V|^{h-1})^{O(1)}$.

Third, Theorems \ref{thm:ccm-direct-weighted} and \ref{thm:ccm} apply to permutation matrices with prescribed
diagonal sums. A permutation matrix of order $N$ is a perfect matching in the
bipartite graph between rows and columns. Coloring an edge $(i,j)$ by its
diagonal class turns diagonal sums into color counts. If only $h$ diagonal
classes have nonzero prescribed sums, the other diagonals can be deleted, and
the resulting bipartite graph has $2N$ vertices, at most $Nh$ edges, and
target matching size $N$. Either Theorem \ref{thm:ccm-direct-weighted} or Theorem~\ref{thm:ccm} therefore gives a randomized
algorithm with running time $N^{O(h)}$, and hence polynomial time for fixed $h$.

We now turn to the proof of the main additive-transversal theorem.

\section{Proof of Theorem~\ref{thm:bounded-support}}
\label{sec:proof-bounded-support}

We now prove Theorem~\ref{thm:bounded-support} by reducing Additive Transversal
to Color-Counted Matching and then applying
Theorem~\ref{thm:ccm-direct-weighted}. The self-contained
Theorem~\ref{thm:ccm} gives an alternative route to randomized polynomial
time for fixed support size, but the sharper bound stated in
Theorem~\ref{thm:bounded-support} comes from
Theorem~\ref{thm:ccm-direct-weighted}.

The reduction works as follows. We consider an instance with a multiset
$A=(a_1,\dots,a_k)$ of elements of $\mathbb Z_m$ and a set
$B\subseteq \mathbb Z_m$ of cardinality $k$. Let
$\supp(A)=\{\alpha_0,\dots,\alpha_{s-1}\}$, and let $m_j$ be the multiplicity
of $\alpha_j$ in $A$.

We build a bipartite graph, called the \textit{diagonal graph} of $A$ and $B$, as
follows. The left vertex set is a copy of $B$, denoted by $L$. The right
vertex set is a disjoint copy $R$ of the set
$\{\alpha_j+b : 0\le j<s,\ b\in B\}\subseteq \mathbb Z_m$.
For every $b\in B$ and every $j\in\{0,\dots,s-1\}$, we add the edge
connecting the left copy of $b$ to the right copy of $b+\alpha_j$, and color
it by $j$.
Thus
the diagonal graph has $k$ left vertices, at most $ks$ right vertices, and exactly $ks$
edges.

\begin{lemma}
\label{lem:additive-to-Color-CountedMatching}
The additive transversal instance with multiset $A$ and set $B$ has a solution
if and only if the diagonal graph built from $A$ and $B$ has a matching using
exactly $m_j$ edges of color $j$ for every $j\in \{0,\dots,s-1\}$.
\end{lemma}

\begin{proof}
We first assume that the additive transversal problem admits a solution. Thus
there is an ordering $b_1,\dots,b_k$ of the elements of $B$ such that the sums
$c_i=a_i+b_i$
are pairwise distinct.

For each index $i$ from $1$ to $k$, let $j(i)$ be the unique index such that
$a_i=\alpha_{j(i)}$. By construction of the diagonal graph, there is an edge
from the left vertex $b_i$ to the right vertex $c_i=b_i+\alpha_{j(i)}$, and
this edge has color $j(i)$. Let
$M=\{\, b_i c_i : 1\le i\le k\,\}$.
Since the vertices $b_1,\dots,b_k$ are precisely the elements of $B$, the
matching $M$ saturates $L$. Moreover, since the sums $c_1,\dots,c_k$ are
pairwise distinct, no two edges of $M$ have the same right endpoint. Hence
$M$ is indeed a matching.

It remains to check the color counts. For each color $j$, the number of edges
of color $j$ in $M$ is exactly the number of indices $i$ such that
$a_i=\alpha_j$. This number is precisely the multiplicity $m_j$ of
$\alpha_j$ in $A$. Therefore $M$ is a matching using exactly
$m_j$ edges of color $j$ for every $j\in\{0,\dots,s-1\}$.

Conversely, assume that the diagonal graph admits a matching $M$ using exactly
$m_j$ edges of color $j$ for every $j\in\{0,\dots,s-1\}$. Since
$\sum_{j=0}^{s-1}m_j=k$, the matching $M$ has exactly $k$ edges, and therefore
saturates the left side $L$.

For each $b\in B$, let $c_b$ be the right endpoint matched to the left copy of
$b$. If the edge $bc_b$ has color $j$, then by construction
$c_b=b+\alpha_j$.

For each color $j$, let $I_j=\{i\in\{1,\dots,k\}: a_i=\alpha_j\}$
and let $B_j$ be the set of elements $b\in B$ whose left copy is matched by
an edge of color $j$. By the prescribed color counts, both $I_j$ and $B_j$
have size $m_j$. Choose an arbitrary bijection from $I_j$ to $B_j$.
Doing this for every $j$ gives a bijection between the index set
$\{1,\dots,k\}$ of $A$ and the set $B$. Equivalently, it gives an ordering
$b_1,\dots,b_k$ of $B$.

If $i\in I_j$, then $a_i=\alpha_j$, and the matched right endpoint of the
left copy of $b_i$ is $b_i+\alpha_j=b_i+a_i$. Since $M$ is a matching, these
right endpoints are pairwise distinct. Therefore the sums $a_i+b_i$ are
pairwise distinct, and the additive transversal instance has a solution.
\end{proof}

We now use Lemma~\ref{lem:additive-to-Color-CountedMatching} and the
Color-Counted Matching algorithm to prove Theorem~\ref{thm:bounded-support}.

\begin{proof}
We consider an additive transversal instance over $\Z_m$, with
$|A|=|B|=k$, and we denote $s=|\supp(A)|$.
By Lemma~\ref{lem:additive-to-Color-CountedMatching}, the instance reduces to
a Color-Counted Matching instance on the diagonal graph, with target counts
$m_0,\dots,m_{s-1}$. Hence the target matching size is
$q=\sum_{j=0}^{s-1}m_j=k$, and the number of colors is $h=s$. 
The diagonal graph has $k$ left vertices, at most $ks$ right vertices, and
exactly $ks$ edges. Applying Theorem~\ref{thm:ccm-direct-weighted} with
$q=k$ and $h=s$ gives randomized search time $(k+1)^{s-1}\poly(k,s)$
in the unit-cost model.

It remains to account for arithmetic in $\Z_m$. Constructing the diagonal
graph requires computing, comparing, and storing the $ks$ residues
$b+\alpha_j\pmod m$, each represented with $O(\log m)$ bits. The matching
weights use $O(s\log(k+1))$ bits. Consequently, the complete bit complexity is $(k+1)^{s-1}\poly(k,s,\log m)$.

If the Color-Counted Matching algorithm returns a matching with the prescribed
color-count vector, Lemma~\ref{lem:additive-to-Color-CountedMatching}
converts it into a valid additive transversal. The returned solution can be
checked deterministically. If a solution exists, the success probability is at
least $2/3$, and it can be amplified by independent repetition. This proves
Theorem~\ref{thm:bounded-support}.
\end{proof}

\section{Proof of Theorem~\ref{thm:ccm}}
\label{sec:self-contained-ccm}

We now provide the more self-contained proof leading to Theorem \ref{thm:ccm}.

\subsection{From $h$ colors to two colors} 

We first reduce Color-Counted Matching to Exact Red Matching.

\begin{definition}
An Exact Red Matching instance consists of a graph $(V,E)$, a set of red
edges $E_{\mathrm{red}}\subseteq E$, and an integer $T$. The task is to find a
perfect matching $M$ such that $|M\cap E_{\mathrm{red}}|=T$.
\end{definition}

The Exact Red Matching problem is the special case of Color-Counted Matching
with two colors, red and black, in which the matching is required to be
perfect. Indeed, every perfect matching
has $\frac{|V|}{2}$ edges. Thus, prescribing the number of red edges of a perfect
matching is equivalent to prescribing target counts $r_{\mathrm{red}}$ and
$r_{\mathrm{black}}$ satisfying
$r_{\mathrm{red}}+r_{\mathrm{black}}=\frac{|V|}{2}$.

\begin{figure}[htb]
\begin{center}
\includegraphics[width=0.6\textwidth]{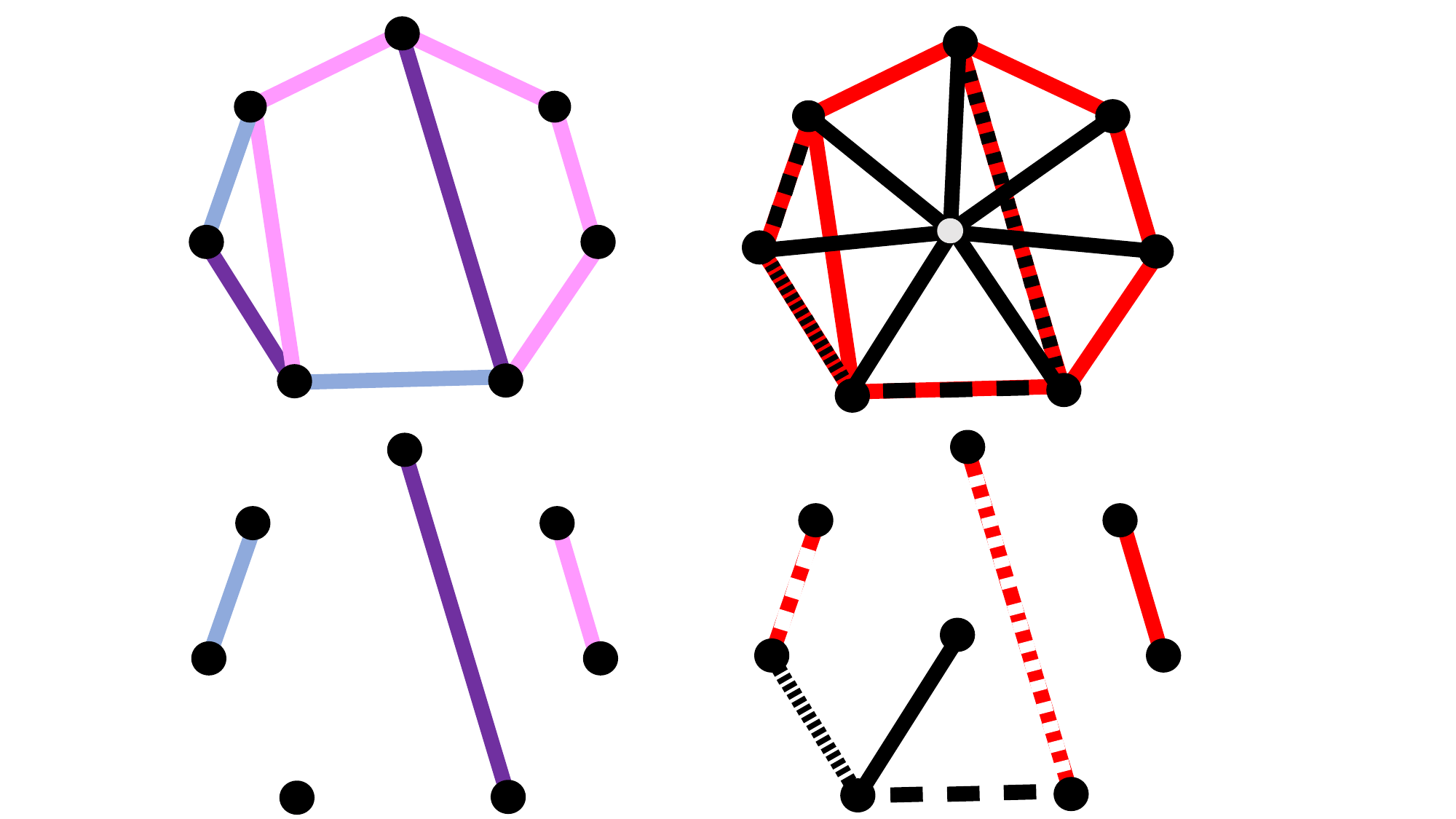}
\end{center}
\caption{\label{red}
\textbf{Reduction of a $3$-color Color-Counted Matching instance with target
counts $(1,1,1)$ to an Exact Red Matching instance with $21$ red edges.}
The Color-Counted Matching instance appears on the left, and the red-black
$Q$-augmented graph appears on the right. The target matching size is
$q=1+1+1=3$, so $Q=q+1=4$, and the red target is
$T=Q^0+Q^1+Q^2=1+4+16=21$. Corresponding solutions are shown below.}
\end{figure}

\paragraph{The red-black $Q$-augmented graph.}

We now describe the construction of the Exact Red Matching instance. Let
$H=(V,E)$ be a Color-Counted Matching instance whose edges are colored with
colors $0,\dots,h-1$, and let $r_0,\dots,r_{h-1}$ be the prescribed target
counts. We denote $q=\sum_{i=0}^{h-1} r_i$. Thus any matching satisfying the
prescribed color counts has size $q$.

If $q>\lfloor \frac{|V|}{2}\rfloor$, then no matching of size $q$ exists in an
$|V|$-vertex graph, and we may immediately report failure. Hence we assume from
now on that $q\le \lfloor \frac{|V|}{2} \rfloor$. If $q=0$, the empty matching is a solution of the Color-Counted
Matching instance. We therefore assume that $q\ge1$ and set $Q=q+1$.

The red-black $Q$-augmented graph $H'=(V',E')$ is built from $H$ as follows.
The construction is illustrated in Fig.~\ref{red}.

\begin{enumerate}
\item We initialize $V'=V$ and $E'=\emptyset$.

 \item Each edge $e=(u,v)$ of color $i$ in $E$ is replaced by a path
 $u=x_{e,0},\ x_{e,1},\ \dots,\ x_{e,2Q^i-1}=v$
 with $2Q^i-1$ edges and $2Q^i-2$ new internal vertices added to $V'$ and
 used only for the edge $e$.

 The $Q^i$ edges $x_{e,2t}x_{e,2t+1}$ with $0\le t<Q^i$ are colored red,
 while the $Q^i-1$ edges $x_{e,2t+1}x_{e,2t+2}$ with $0\le t<Q^i-1$ are
 colored black. All these path edges are added to $E'$.

 \item If $2q<|V|$, then we add to $V'$ a set $D$ of $|V|-2q$ dummy vertices
 and join every vertex of $D$ to every original vertex of $V$ by a black edge.
 These dummy edges are added to $E'$.
\end{enumerate}

The target number of red edges in the Exact Red Matching instance is
$T=\sum_{i=0}^{h-1} r_iQ^i$.

\paragraph{Size of the red-black $Q$-augmented graph.}

We now compute the size of the red-black $Q$-augmented graph. For an edge
$e\in E$ of color $\chi(e)$, the replacement path has $2Q^{\chi(e)}$ vertices
and $2Q^{\chi(e)}-1$ edges. Since the two endpoints of the path are the
original endpoints of the edge, this replacement adds $2Q^{\chi(e)}-2$ new
vertices. Therefore
$|V'|=|V|+(|V|-2q)+\sum_{e\in E}(2Q^{\chi(e)}-2)$, and hence
$|V'|\le 2|V|+2|E|Q^{h-1}$.

Similarly, the path replacements contribute
$\sum_{e\in E}(2Q^{\chi(e)}-1)$ edges. The dummy vertices contribute
$|V|(|V|-2q)$ black edges, since each dummy vertex is joined to every original
vertex. Hence
$|E'|=|V|(|V|-2q)+\sum_{e\in E}(2Q^{\chi(e)}-1)$, and therefore
$|E'|\le |V|^2+2|E|Q^{h-1}$.

Thus the red-black $Q$-augmented graph has size
$|H'|=O\!\left(|V|^2+|E|Q^{h-1}\right)$. Since $Q=q+1$, this is
$|H'|=O\!\left(|V|^2+|E|(q+1)^{h-1}\right)$.

\begin{remark}
The construction of the red-black $Q$-augmented graph is not necessarily
optimal. One could use mixed radices instead of the uniform base $Q$. More
precisely, if $Q_i$ is a strict upper bound on the number of selected edges of
color $i$, one may define $W_0=1$ and
$W_i=\prod_{j=0}^{i-1}Q_j$ for $i\ge 1$,
and replace the weight $Q^i$ by
$W_i$. A preprocessing step could further reduce the size of the
augmented graph by taking the target counts $r_i$ into account. We do not use
this optimization here.
\end{remark}

\paragraph{Equivalence between $h$ colors and two colors in the $Q$-augmented graph.}

We now show the equivalence between the Color-Counted Matching instance in
$H$ and the Exact Red Matching instance in the red-black $Q$-augmented graph
$H'$.

\begin{proposition}\label{prop}
Let $H=(V,E)$ be a Color-Counted Matching instance with colors
$0,\dots,h-1$ and target color counts $r_0,\dots,r_{h-1}$. Let
$q=\sum_{i=0}^{h-1}r_i$, assume that $1\le q\le \lfloor \frac{|V|}{2}\rfloor$, and set
$Q=q+1$. Let $H'=(V',E')$ be the red-black $Q$-augmented graph constructed above.
Then $H$ has a matching with the prescribed color counts if and only if $H'$
has a perfect matching with exactly $T=\sum_{i=0}^{h-1} r_iQ^i$ red edges.
\end{proposition}

\begin{proof}
We first prove that, given a solution $M$ of the Color-Counted Matching
instance, we can build a solution of the Exact Red Matching instance for the
corresponding red-black $Q$-augmented graph with exactly
$T=\sum_{i=0}^{h-1} r_iQ^i$ red edges. Thus $M$ contains exactly $r_i$ edges
of color $i$ for every $i\in\{0,\dots,h-1\}$, and therefore $|M|=q$.

We build a matching $M'$ in $H'$ as follows. For each edge $e=uv$ of $M$, we
select in $M'$ all the red edges of the path replacing $e$. For each edge
$e=uv$ of $E$ that does not belong to $M$, we select in $M'$ all the black
edges of the path replacing $e$. Finally, we match the $|V|-2q$ vertices of
$V$ that are not matched by $M$ bijectively to the $|V|-2q$ dummy vertices.

This gives a perfect matching of $H'$. Indeed, in every path replacing an edge
of $H$, all internal vertices are matched. If the corresponding edge belongs
to $M$, then the two original endpoints are matched through the red edges of
the path. If it does not belong to $M$, then the black edges of the path match
all internal vertices and leave the two original endpoints available. These
available original vertices are either matched through another selected path
or, if they are unmatched in $M$, matched to dummy vertices.

The red edges selected in $M'$ come only from the paths corresponding to the
edges of $M$. Each selected edge of color $i$ contributes exactly $Q^i$ red
edges. Hence the total number of red edges in $M'$ is
$\sum_{i=0}^{h-1} r_iQ^i=T$. Thus $M'$ is a perfect matching of $H'$ with
exactly $T$ red edges.

We now prove the converse implication. Assume that $H'$ admits a perfect
matching $M'$ with exactly $T$ red edges. We build from it a matching $M$ in
$H$.

For each edge $e=uv$ of $H$, consider the alternating red-black path replacing
$e$ in $H'$. Since the internal vertices of this path have no incident edges
outside the path, a perfect matching of $H'$ has only two possible states on
this path: either it contains all the red edges of the path, or it contains
all the black edges of the path. We call the path active in the first case
and inactive in the second case. We select the edge $e$ in $M$ if and only if
the path replacing $e$ is active.

It follows that $M$ is a matching in $H$. Indeed, two active paths cannot
share an original endpoint, since otherwise $M'$ would contain two edges
incident with the same vertex.

We now observe that $M$ has exactly $q$ edges. In the graph $H'$, every dummy
vertex is matched to an original vertex. Since there are $|V|-2q$ dummy
vertices, exactly $|V|-2q$ original vertices are matched to dummy vertices.
The remaining $2q$ original vertices must be matched through active paths.
Each active path uses exactly two original endpoints. Hence exactly $q$ paths
are active, and therefore $|M|=q$.

For each $i\in\{0,\dots,h-1\}$, let $y_i$ be the number of selected edges of
color $i$ in $M$. Equivalently, $y_i$ is the number of active paths of weight
$Q^i$ in $H'$. Since an active path of color $i$ contributes exactly $Q^i$
red edges, the number of red edges in $M'$ is
$\sum_{i=0}^{h-1} y_iQ^i$. Since $M'$ has exactly $T$ red edges, $\sum_{i=0}^{h-1} y_iQ^i=\sum_{i=0}^{h-1} r_iQ^i$.

We now use uniqueness of base-$Q$ expansion. Since $|M|=q$ and $Q=q+1$, we
have $y_i<Q$ for every $i$. Also, since $r_i\le q$ and $Q=q+1$, we have
$r_i<Q$ for every $i$. Thus both sides of the previous equality are base-$Q$
expansions with digits strictly smaller than $Q$. Consequently, for every
$i\in\{0,\dots,h-1\}$, we have $y_i=r_i$.

Therefore $M$ contains exactly $r_i$ edges of color $i$, and so
$M$ is a solution of the original Color-Counted Matching instance.
\end{proof}

\subsection{Proof of Theorem~\ref{thm:ccm}}

The original result of Mulmuley--Vazirani--Vazirani places exact matching in
randomized NC. In this paper we only use the resulting randomized
polynomial-time search algorithm that we restate here.

\begin{theorem}[\cite{mulmuley1987matching}]
\label{thm:mvv-exact-matching}
Let $(V,E)$ be a graph, let $E_{\mathrm{red}}\subseteq E$, and let $T\ge 0$.
There is a randomized algorithm running in time polynomial in $|V|+|E|$ which
either reports failure or outputs a perfect matching $M$ such that
$
 |M\cap E_{\mathrm{red}}|=T.
$
If such a perfect matching exists, then the algorithm outputs one with
probability at least $2/3$. If no such perfect matching exists, then the
algorithm reports failure with probability $1$.
The success probability on yes-instances can be increased to $1-\varepsilon$,
for any $\varepsilon>0$, by independent repetition, with an additional
multiplicative factor $O(\log(1/\varepsilon))$ in the running time.
\end{theorem}

The original theorem of Mulmuley--Vazirani--Vazirani applies to general graphs.
Applying Theorem~\ref{thm:mvv-exact-matching} to the $Q$-augmented
red-black graph $H'$ gives a randomized algorithm whose running time is polynomial
in $|V'|+|E'|$. We now prove Theorem~\ref{thm:ccm}.

\begin{proof}[Proof of Theorem~\ref{thm:ccm}]
Let $q=\sum_{i=0}^{h-1}r_i$. If $q=0$, the algorithm returns the empty matching. We therefore assume that $q\ge1$.
If $q>\lfloor \frac{|V|}{2}\rfloor$, then no matching of
size $q$ exists in $(V,E)$, and the algorithm reports failure. We therefore assume
$q\le \lfloor \frac{|V|}{2}\rfloor$, and set $Q=q+1$.

We construct the red-black $Q$-augmented graph $H'=(V',E')$ as in
Proposition~\ref{prop}, and set the red target to be
$T=\sum_{i=0}^{h-1}r_iQ^i$.
By Proposition~\ref{prop}, the original Color-Counted Matching instance has a solution if and only
if $H'$ has a perfect matching with exactly $T$ red edges.

We now apply the exact matching theorem of Mulmuley--Vazirani--Vazirani to the
red-black graph $H'$ with red-edge set $E'_{\mathrm{red}}$ and target $T$, where $E'_{\mathrm{red}}$ is the set of red edges of $E'$.
If such a perfect matching exists, the MVV algorithm outputs one with
probability at least $2/3$; if none exists, it reports failure with probability
$1$. Any returned perfect matching can be checked deterministically.

The running time for building the red-black $Q$-augmented graph is linear in its size, namely $
 |V'|+|E'| = O\!\left(|V|^2+|E|Q^{h-1}\right)$.
The MVV algorithm runs in time polynomial in the size of its input, so the
total running time is $ \left(|V|^2+|E|Q^{h-1}\right)^{O(1)}.$
Since $Q=q+1$, this is
$\left(|V|^2+|E|(q+1)^{h-1}\right)^{O(1)}$.

Finally, if the MVV algorithm returns a perfect matching of $H'$ with exactly
$T$ red edges, the constructive correspondence in Proposition~\ref{prop} allows us to extract the active path gadgets from it. These active gadgets form a matching
of $H$ with color-count vector $(r_0,\dots,r_{h-1})$. Hence the algorithm
outputs a correct Color-Counted Matching solution. The failure probability can be reduced by
independent repetition.
\end{proof}

\section*{Acknowledgments}
We thank Alexandra Lassota for drawing our attention to the ICALP
2022 exact-matching result~\cite{lassotaLukasiewiczPolak2022exactMatching}.
Research reported in this paper was partially supported by the Berlin
Mathematics Research Center MATH$^+$ (EXC-2046/2, project ID 390685689),
funded by the Deutsche Forschungsgemeinschaft (DFG, German Research
Foundation) under Germany's Excellence Strategy as well as by the Natural Sciences and Engineering Research Council of Canada Discovery (Grant program number RGPIN-2026-06553).

\bibliographystyle{plainnat}
\bibliography{references}

 \end{document}